\documentclass[conference]{IEEEtran}

\usepackage{amssymb,amsmath,amsfonts,bm,epsfig,graphicx,theorem,epstopdf,wasysym}
\usepackage{rotating,setspace,latexsym,epsf,color,dsfont,caption}
\usepackage{cite,subfigure}

\newtheorem{Theorem}{Theorem}

\newtheorem{Lemma}{Lemma}

\newtheorem{Definition}{Definition}

\newenvironment{Proof}[1]{\medskip\par\noindent
{\bf Proof:\,}\,#1}{{\mbox{\,$\blacksquare$}\par}}

\newenvironment{proof}[1]{\medskip\par\noindent
{\bf Proof:\,}\,#1}{{\mbox{\,$\blacksquare$}\par}}

\newcommand{\Eb}{\mathbb{E}}
\newcommand{\Pb}{\mathbb{P}}

\ifodd 1
\newcommand{\yjc}[1]{{\color{magenta}(JY: #1)}}
\else
\newcommand{\yjc}[1]{}
\fi

\ifodd 1
\newcommand{\yj}[1]{{\color{red}(JY: #1)}}
\else
\newcommand{\yj}[1]{}
\fi

\ifodd 1
\newcommand{\fst}[1]{{\color{blue}#1}}
\else
\newcommand{\fst}[1]{}
\fi

\title{Optimal Status Updating for an Energy Harvesting Sensor with a Noisy Channel}
\author{
 \IEEEauthorblockN{Songtao Feng}
    \IEEEauthorblockA{Department of Electrical Engineering \\
The Pennsylvania State University\\
 University Park, PA 16802\\
    \emph{sxf302@psu.edu}}
    \and

   \IEEEauthorblockN{Jing Yang}
    \IEEEauthorblockA{Department of Electrical Engineering \\
The Pennsylvania State University\\
 University Park, PA 16802\\
    \emph{yangjing@psu.edu}}

\thanks{This work was supported in part by the National Science Foundation (NSF) under Grant ECCS-1650299.}
}

\begin{document}
\IEEEoverridecommandlockouts
\maketitle
\thispagestyle{empty}

\begin{abstract}
Consider an energy harvesting sensor continuously monitors a system and sends time-stamped status update to a destination. The destination keeps track of the system status through the received updates. Under the energy causality constraint at the sensor, our objective is to design an optimal {\it online} status updating policy to minimize the long-term average Age of Information (AoI) at the destination. We focus on the scenario where the the channel between the source and the destination is noisy, and each transmitted update may fail independently with a constant probability. We assume there is no channel state information or transmission feedback available to the sensor. We prove that within a broadly defined class of online policies, the best-effort uniform updating policy, which was shown to be optimal when the channel is perfect, is still optimal in the presence of update failures. Our proof relies on tools from Martingale processes, and the construction of a sequence of {\it virtual} policies.

\end{abstract}


\section{Introduction}


Energy harvesting (EH) sensor networks, composed of devices that are powered by energy harvested from ambient environment, are becoming the future of energy self-sustaining wireless networks, with the goal of having extended lifetime and being deployed in challenging conditions or locations. To cope with the intermittent, random and scarce nature of the harvested energy in such networks, various energy management policies have been studied in the past years under different performance criterion~\cite{Yang_tcom,ozel_finite_tcom,HoZ12,kaya_tcom}.

Meanwhile, a metric called ``Age of Information'' (AoI) has been introduced to measure the timeliness of the status information in a network recently~\cite{infocom/KaulYG12}. Specifically, at time $t$, the AoI in the system is defined as $t-U(t)$, where $U(t)$ is the time stamp of the latest received update packet at the destination. AoI has shown to be fundamentally different from standard performance metrics, such as throughput, delay, or distortion. Modeling the status updating system as a queueing system, the time average AoI has been analyzed in systems with a single server~\cite{infocom/KaulYG12,ciss/KaulYG12,isit/YatesK12,YatesK16,Pappas:2015:ICC,isit/NajmN16,isit/KamKNWE16,isit/ChenH16}, and multiple servers \cite{isit/KamKE13, isit/KamKE14,tit/KamKNE16}. A related metric, Peak Age of Information (PAoI), has been introduced and studied in \cite{isit/CostaCE14,tit/CostaCE16,isit/HuangM15}. The optimality properties of a preemptive Last Generated First Served service discipline are identified in \cite{isit/BedewySS16}. AoI optimization has been studied in~\cite{infocom/SunUYKS16}.
The relationship between AoI and the MMSE in remote estimation of a Wiener process is investigated in \cite{Sun:ISIT:2017}.

A few recent works start to investigate AoI-minimal status updating policies under an energy harvesting setting~\cite{isit/Yates15,ita/BacinogluCU15,Yang:2017:AoI,BacinogluU17,Ahmed:2018:ICC,Ahmed:2018:ITA,Uysal:2018:EH,Ahmed:2017:Asilomar,Ahmed:2017:Globecom,Baknina:2018:CISS,Baknina:2018:ISIT}. It has been shown in \cite{isit/Yates15} that a {\it lazy} updating policy that introduces inter-update delays outperforms a greedy policy that submits a fresh update as the system becomes available. AoI minimization in an EH system under different assumptions on the battery size has been investigated in \cite{ita/BacinogluCU15,Yang:2017:AoI,BacinogluU17,Ahmed:2018:ICC,Ahmed:2018:ITA,Uysal:2018:EH}. Specifically, for the infinite battery case, \cite{Yang:2017:AoI} shows that the 
{\it best-effort uniform} (BU) updating policy, which updates at a constant rate when the source has sufficient energy, is optimal when the channel between source and destination is perfect. For finite battery sizes, several battery level dependent threshold policies have been shown to be optimal in \cite{ita/BacinogluCU15,Yang:2017:AoI,BacinogluU17,Ahmed:2018:ICC,Ahmed:2018:ITA,Uysal:2018:EH}. Offline policies to minimize AoI in EH channels have been studied in \cite{Ahmed:2017:Asilomar,Ahmed:2017:Globecom}. Average AoI with different channel coding schemes for EH channels has been analyzed in \cite{Baknina:2018:CISS,Baknina:2018:ISIT}.  

In this paper, we extend our previous work \cite{Yang:2017:AoI} by assuming a noisy channel between the source and the destination, and each update will fail with a constant probability, independent with any other factors in the system. Besides, we assume there is no channel state information (CSI) or transmission feedback available to the transmitter. We focus on the infinite battery case, and aim to develop online status updating policy for the source to minimize the long-term average AoI at the destination even with the noisy channel. We first obtain a lower bound on the long-term AoI for a broadly defined class of online policies, and then show that the BU updating policy can actually achieve the lower bound, thus is still optimal. To overcome the difficulty of characterizing the complicated AoI evolution due to update failure and battery outage, we construct a sequence of virtual policies named as {\it best-effort updating with energy removal} (BU-ER). Under BU-ER, we are able to decouple the impacts of battery outage and update failure, and explicitly show that the expected long-term average AoI under such policies approaches the lower bound. Since the BU-ER policies are sub-optimal to the BU updating, the optimality of BU updating can thus be proved. We also evaluate the performances of the proposed policies through simulations.


\section{System Model and Problem Formulation} \label{sec:model}
Consider a scenario where an energy harvesting sensor continuously monitors a system and sends time-stamped status updates to a destination. The destination keeps track of the system status through the received updates. We use the metric Age of Information (AoI) to measure the ``freshness" of the status information available at the destination.

Similarly to~\cite{Yang:2017:AoI,BacinogluU17,Ahmed:2018:ICC,Ahmed:2018:ITA,Uysal:2018:EH,Ahmed:2017:Asilomar,Ahmed:2017:Globecom,Baknina:2018:CISS,Baknina:2018:ISIT}, we assume the time used to collect and transmit a status update is negligible compared with the time scale of inter-update delays. Therefore, given sufficient energy available at the source, a status update can be generated and transmitted to the destination instantly. Intuitively, a status update should be transmitted once it is generated to avoid unnecessary queueing delay. We assume the channel between the source and the destination is noisy, thus each update transmitted by the source may be corrupted and unrecognizable at the destination. Specifically, we assume that with probability $p$, $0<p\leq 1$, an update will be successfully delivered to the destination, irrespective of other factors in the system. As shown in Fig. \ref{fig:AoI}, only after an update is successful received the AoI in the system will be updated.
We assume there is no CSI or transmission feedback available to the source. Therefore, the source does not have precise knowledge of the instantaneous AoI in the system.

We assume that the energy unit is normalized so that each status update requires one unit of energy. This energy unit represents the cost of both measuring and transmitting a status update. Assume energy arrives at the sensor according to a Poisson process with parameter $\lambda$. Hence, energy arrivals occur at discrete time instants $t_1,t_2,\ldots$. Without loss of generality, we assume $\lambda=1$ for ease of exposition. The sensor is equipped with a battery to store the harvested energy. In this paper, we focus on the case when battery size is infinite.

A status update policy is denoted as $\pi:=\{l_n\}_{n=1}^\infty$, where $l_n$ is the $n$th updating epoch at the {\it source}.
Define $A_n$ as the total amount of energy harvested in $[l_{n-1},l_n)$, and $E(l^-_n)$ as the energy level of the sensor right before the scheduled updating epoch $l_n$.
Assume $l_0=0$,
\begin{align}
E(l_0^-)=E_0, \mbox{ where } E_0\geq 1.
\label{eqn:energy_initial}
\end{align}
Then, under any feasible status update policy, the energy queue evolves as follows
\begin{align}
E(l^-_{n})&=E(l^-_{n-1})-1+A_n, \label{eqn:energy_queue}\\
E(l_n^-)&\geq 1, \label{eqn:energy_constraint}
\end{align}
for $n=1,2,\ldots$. Equation (\ref{eqn:energy_constraint}) corresponds to the energy causality constraint in the system.
Based on the Poisson arrival process assumption, $A_n$ is an independent Poisson random variable with parameter $l_{n}-l_{n-1}$. 

Due to channel fading, only a subset of the updates will be successfully delivered. Thus, the actual status updating epochs at the {\it destination} are different from $\{l_n\}_{n=1}^\infty$ in general. We use $S_n$ to denote the $n$th actual update epoch at the {\it destination}.   
We assume $S_0=l_0=0$, i.e., the system successfully updates its status information right before time zero. 

Denote the inter-update delays as $X_n:= S_n-S_{n-1}$, for $n=1,2,\ldots$. Then, we have $S_n=\sum_{i=1}^nX_i$.
We use $M(T)$ and $N(T)$ to denote the number of transmitted status updates and successfully delivered status updates over $(0,T]$, respectively. Define $R(T)$ as the accumulated age of information experienced by the system over $[0,T]$. Then,
\begin{align}\label{defn:R(T)}
R(T)&=\frac{\sum_{i=1}^{N(T)}X_i^2+ (T-S_{N(T)})^2}{2},
\end{align}
and the time average AoI over the duration $[0,T]$ can be expressed as $R(T)/T$.

\begin{figure}[t]
\centering
\includegraphics[scale=0.65]{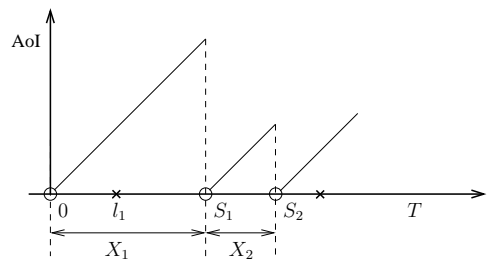}
\caption{AoI as a function of $T$. Circles represent successful status updates and crosses represent failed status updates.} \label{fig:AoI}
\vspace{-0.1in}
\end{figure}

Our objective is to determine the sequence of update epochs $l_1,l_2,\ldots$ at the {\it source}, so that the time average AoI at the {\it destination} is minimized, subject to the energy causality constraint. We focus on a set of {\it online} policies $\Pi$ in which the information available for determining the updating epoch $l_n$ includes the updating history $\{l_i\}_{i=0}^{n-1}$, the energy arrival profile over $[0,l_n)$, as well as the energy harvesting statistics (i.e., $\lambda$ in this scenario).
The optimization problem can be formulated as
\begin{eqnarray}\label{eqn:opt}
\underset{\pi\in\Pi}{\min} & &\limsup_{T\rightarrow \infty} \Eb\left[\frac{R(T)}{T}\right]\\
\mbox{s.t. } & & (\ref{eqn:energy_initial})-(\ref{eqn:energy_constraint}),\nonumber
\end{eqnarray}
where the expectation in the objective function is taken over all possible energy harvesting sample paths.

Due to the stochastic energy arrivals and temporal depending in the battery state, it is difficult to solve the stochastic optimization in (\ref{eqn:opt}) directly. The random update failures make the problem even more challenging. Therefore, in the following, we take an indirect approach, where we will first identify a lower bound on the long-term AoI for a broad class of online polices, and then construct online policies to achieve the lower bound.

\section{A Lower Bound}  \label{sec:lower}
First, we note that when the battery size is infinite, no energy overflow will happen, and the long-term average status updating rate is subject to the EH rate constraint. Specifically, we have the following lemma.

\begin{Lemma}[Lemma 1 in \cite{Yang:2017:AoI}]\label{lemma:rate}
Under any policy $\pi\in\Pi$, it must have $\lim_{T\rightarrow \infty} M(T)/T \leq 1$ almost surely.
\end{Lemma}

Besides, we also have the following intuitive yet important observation. 
\begin{Lemma}\label{lemma:finiteAoI}
For any $\pi\in\Pi$ that achieves a {\it finite} expected long-term average AoI, it must have $\lim_{T\rightarrow \infty} M(T)=\infty$ almost surely.
\end{Lemma}

\if{0}
\begin{proof}
We prove it through contradiction. If $$\Pb[\lim_{T\rightarrow \infty} M(T)=\infty]<1,$$
there exists $\epsilon>0$ and $M_0>0$, such that 
$$\Pb[\lim_{T\rightarrow \infty} M(T)<M_0]\geq \epsilon.$$
Define
\begin{align} \label{defn:p_n}
p_n:=(1-p)^{n-1}p,
\end{align}
 i.e., the probability that $l_n$ is the first successful update after $l_0$.
Then, 
\begin{align}
&\limsup_{T\rightarrow \infty} \Eb\left[\frac{R(T)}{T}\right]\\
&\geq \lim_{T\rightarrow \infty} \frac{T^2}{2T}\cdot\Pb[\mbox{all $M(T)$ updates fail}, M(T)<M_0]\\
&\geq \lim_{T\rightarrow \infty} \frac{T}{2}\left(1-\sum_{i=1}^{M_0}p_i\right) \epsilon=\infty
\end{align}
which implies that the expected long-term average AoI cannot be finite. 
\end{proof}
\fi

The proof of Lemma~\ref{lemma:finiteAoI} is omitted due to space limitation.

In the following, we will focus on the policies that achieve finite expected long-term average AoI.
In order to facilitate our analysis, we introduce a broad class of online policies defined as follows.
\begin{Definition}[Bounded Updating Policy]
	If under a policy $\pi\in\Pi$, the $n$th updating epoch at the source (i.e., $l_n$) satisfies $\Eb[l_n]<\infty$ for any fixed $n\in\{1,2,\ldots\}$, $\pi$ is called a {\it bounded updating} policy.
\end{Definition}

Denote the set of bounded updating policy as $\Pi'$. Then, $\Pi'\subset\Pi$. Intuitively, any practical status updating policy should be in $\Pi'$, as it is undesirable to have any $n$th updating epoch, and the inter-update delay between any consecutive updating epochs before $l_n$, to become unbounded {\it in expectation}. We have the following lower bound for bounded updating policies.
\begin{Lemma}\label{lemma:lower}
The expected long-term average AoI is lower bounded by $\frac{2-p}{2p}$ for any $\pi\in\Pi'$.
\end{Lemma}
\begin{Proof}
Define $S_{i}^{T}:=\min\{S_i,T\}$, $l_{n}^{T}:=\min\{l_n,T\}$, and $p_n:=(1-p)^{n-1}p$. Then, under any $\pi\in\Pi'$, the expected average AoI over $[0,T]$ can be expressed as
\begin{align} 
\Eb\left[\frac{R(T)}{T}\right]  
&=\frac{1}{T}\Eb\left[\sum_{i=0}^{N(T)}\frac{(S_i^T - S_i)^2}{2}\right] 
\label{eqn:lowbound-11} \\
&=\frac{1}{2T}\Eb\left[ \sum_{n=1}^{M(T)} p_n l_n^2+\left(1-\sum_{n=1}^{M(T)}p_n\right) T^2\right.\nonumber\\
&\quad\quad+\left. \sum_{n=1}^{M(T)} \sum_{j=1}^{\infty} (l_{n+j}^T - l_n)^2p p_j 
\right], \label{eqn:lowbound-12}
\end{align}
where the first two terms inside the expectation in (\ref{eqn:lowbound-12}) correspond to the AoI contribution over $[0,S^T_1]$, and the last term correspond to the AoI contribution over any other $[S_i,S^T_{i+1}]$. This can be explained as follows. With fixed updating epochs $\{l_n\}$, depending on the realization of the channel state, the interval $[0,T]$ can be decomposed into segments, separated by successful updates. The probability to have $[l_n, l^T_{n+j}]$, $1\leq n\leq M(T),j\geq 1$, as one of such segment equals $pp_j$, which corresponds to the event that update at $l_n$ succeeds, and the next successful update is at $l_{n+j}$. The corresponding AoI contribution over $[l_n, l^T_{n+j}]$ thus needs to be weighted by $pp_j$ when the expected AoI is calculated. Since the AoI contribution over $[0,S^T_1]$ is always positive, in the following, we will drop it to obtain a lower bound, i.e., 
\begin{align}
&\lim_{T\rightarrow \infty}\Eb\left[\frac{R(T)}{T}\right]  \\
&\geq \lim_{T\rightarrow \infty}\frac{1}{2T} \Eb\left[ p\sum_{j=1}^{\infty}p_j \sum_{n=1}^{M(T)} (l_{n+j}^{T}-l_n)^2  \right]\\
& \geq \lim_{T\rightarrow \infty} \frac{1}{2T} \Eb\left[ p\sum_{j=1}^{\infty}p_j \frac{1}{M(T)} \left(jT-\sum_{n=1}^{j} l^T_n\right)^2  \right]\label{eqn:lowbound-3}  \\
&
= \lim_{T\rightarrow \infty} \frac{1}{2}p\sum_{j=1}^{\infty}p_{j} j^2 \Eb\left[\frac{(T-\bar{l}_j^T)^2}{M(T)T} \right],\label{eqn:lowbound-31}
\end{align}
where (\ref{eqn:lowbound-3}) is based on Jensen's inequality, (\ref{eqn:lowbound-31}) is derived by considering the cases $j\leq M(T)$ and $j>M(T)$ separately, and $\bar{l}_j^T:=\sum_{n=1}^{j}l_n^T/j$.

Since each term in the summation in (\ref{eqn:lowbound-31}) is positive, we can switch the order of limit and summation. We note that for any given $j$, $\Eb[\bar{l}_j^T]\leq \Eb[l_j]<\infty$ according to the definition of bounded policy. Besides, for any policy that renders a finite expected average AoI, we must have $\lim_{T\rightarrow \infty}M(T)=\infty$ almost surely according to Lemma~\ref{lemma:finiteAoI}. Therefore,  according to the Bounded Convergence Theorem, we have
\begin{align}
\lim_{T\rightarrow \infty}\Eb\left[\frac{\bar{l}^T_j}{M(T)}\right]=0, \quad \lim_{T\rightarrow \infty}\Eb \left[\frac{\bar{l}_j^2}{M(T)T} \right]=0 .
\end{align} 
Combining with (\ref{eqn:lowbound-31}), we have
\begin{align}
&\lim_{T\rightarrow \infty}\Eb\left[\frac{R(T)}{T}\right]\geq 
\frac{1}{2}p\sum_{j=1}^{\infty}p_j j^2 \lim_{T\rightarrow \infty} \Eb\left[\frac{T}{M(T)}\right]\\
& =\frac{1}{2}p \sum_{j=1}^{\infty} j^2 (1-p)^{j-1}p  =\frac{2-p}{2p}, \label{eqn:lowbound-4}
\end{align}
where (\ref{eqn:lowbound-4}) follows from Lemma~\ref{lemma:rate}.
\end{Proof}

\section{Optimal Online Status Updating}
In this section, we propose online status updating policies to achieve the lower bound derived in Section~\ref{sec:lower}. We will start with the BU updating policy introduced in \cite{Yang:2017:AoI}. Although we assume a noisy channel in this work, when there is no CSI or feedback available to the source, intuitively, it is still desirable for the source to update in a uniform fashion, so that the successfully received updates at the destination would be most uniformly distributed in time. 


\begin{Definition}[BU Updating]
The sensor is scheduled to update the status at $s_n=n$, $n=1,2,\ldots$. The sensor performs the task at $s_n$ if $E(s^-_n)\geq 1$; Otherwise, the sensor keeps silent until the next scheduled status update epoch.
\end{Definition}

BU updating ensures that the energy causality constraint is always satisfied. We expect that BU updating achieves the lower bound in Lemma~\ref{lemma:lower}, however, analyzing its AoI performance is very challenging. Although we are able to identify a renewal structure in the system status evolution under the BU updating policy (i.e., a renewal interval can begin right after the sensor successfully delivers an update and the battery state becomes $E_0-1$), the analysis of the expected average AoI over one renewal interval is still very complicated, mainly due to two reasons: 

First, different from the perfect channel case~\cite{Yang:2017:AoI}, the actual update epoch at the destination may deviate from the scheduled update epochs $s_n$ due to two possible events: battery outage and update failure. Although the average AoI can be characterized in systems where only one of such events can happen, it is hard to analyze the AoI when the effects of both events are involved. 

Second, the expected length of such a renewal interval is unbounded. This is because the battery evolution under BU updating can be modeled as a Martingale process, and as we will show in the proof of Lemma~\ref{lemma:N(T1)}, the expected time when it becomes empty for the first time (i.e., hitting time of zero) is infinity. Since with a non-zero probability the renewal interval contains such an interval, the expected length of each renewal interval is thus unbounded, and the corresponding expected average AoI becomes intractable.

To overcome such challenges, we will construct a sequence of {\it virtual} policies, and show that the expected time average AoI under those virtual policies approaches the lower bound in Lemma~\ref{lemma:lower}. Since such virtual policies are sub-optimal to the BU updating policy, the optimality of BU updating can thus be proved. In order to simplify the definition and analysis of the virtual policy, we assume $E_0=2$. The proof can be slightly modified to show that the optimality of the proposed policy is valid for any $E_0\geq 0$. 

\begin{Definition}[BU-ER$_{T_0}$] \label{def:bu-er}
The sensor performs BU updating until the battery level after updating, i.e., $E(s_n^+)$, becomes zero for the first time, or until time $T_0^+$, in which case the sensor depletes its battery; After that, when the battery level $E(s_n^+)$ becomes higher than or equal to one for the first time, the sensor reduces $E(s_n^+)$ to one, and then repeats the process.
\end{Definition}

\begin{Lemma} \label{lemma:suboptimal}
For any $T_0>0$, BU-ER$_{T_0}$ updating policy is sub-optimal to the BU updating policy. 
\end{Lemma} 
\begin{proof}
We note that BU-ER$_{T_0}$ updating is identical to BU updating except the energy removal at time $T_0$ and when $E(s_n^+)$ becomes higher than one. Given the same energy harvesting sample path, the battery level under BU is always higher than that under BU-ER$_{T_0}$. Thus, BU-ER$_{T_0}$ incurs more infeasible status updating epochs. With the same channel fading profile, the instantaneous AoI under BU-ER$_{T_0}$ updating is always greater than or equal to that under BU updating sample path-wisely. Thus, the expected time-average AoI under BU-ER$_{T_0}$ is greater than or equal to that under BU, which proves the lemma. 
\end{proof}

Since BU-ER$_{T_0}$ updating policy is a {\it renewal} policy, to analyze the expected long-term average AoI, it suffices to analyze the expected average AoI over one renewal interval. In the following, we will focus on the first renewal interval, and show that the corresponding expected average AoI converges to the lower bound in Lemma~\ref{lemma:lower} as $T_0$ increases. First, as illustrated in Fig~\ref{fig:bu-er}, we note that the renewal interval consists of two stages. The first stage starts at time zero and ends until $E(s_n^+)$ becomes zero for the first time, or until time $T_0^+$. We denote $T_1$ as the duration of the first stage. We note that all scheduled status updating epochs over $(0, T_1]$ are feasible. The second stage starts at $T_1$ and ends when the battery level $E(s_n^+)$ becomes higher than or equal to one for the first time after $T_1$. We denote $T_2$ as the duration of the second stage. 
  
  \begin{figure}[t]\centering
\includegraphics[width=3.3in]{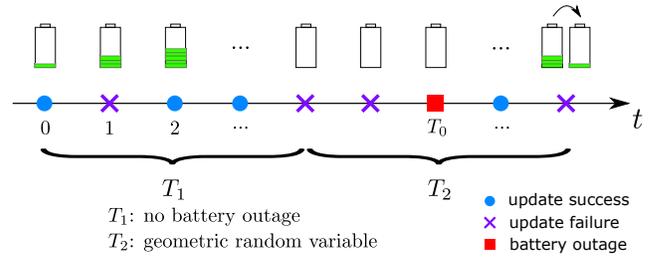}
\vspace{-0.05in}
\caption{An illustration of the BU-ER$_{T_0}$ updating policy.}\label{fig:bu-er}
\vspace{-0.1in}
\end{figure}

\begin{Lemma} \label{lemma:N(T1)}
Under BU-ER$_{T_0}$ updating, $\lim_{T_0\rightarrow \infty} \Eb[T_1]=\infty.$
\end{Lemma}
\begin{proof}
Consider a ``random walk" $\{\Omega_n\}_{n=0}^{\infty}$, which start with $1$ and increments with $A_n-1$, where $A_n$ is an i.i.d. Poisson random variable with parameter $1$. Denote the first $0$-hitting time for $\{\Omega_n\}_{n=0}^{\infty}$ as $\kappa$. Then $\Omega_0=1$ and $\Omega_{\kappa}=0$. Note that when $T_0\rightarrow\infty$, $\{\Omega_n\}_{n=0}^{\kappa}$ is identical to the battery level evolution process $\{E(s_n^+)\}_{n=0}^{\kappa}$ under the BU-ER$_{T_0}$ updating policy almost surely, and the corresponding $T_1=\kappa$. 

Define a Martingale process associated with $\{\Omega_n\}_{n=0}^{\infty}$ as $\{\exp(-\alpha \Omega_n - n\gamma(\alpha))\}_{n=0}^{\infty}$ with $\alpha>0$ and $\gamma(\alpha)=e^{-\alpha}-(1-\alpha)>0$. According to the proof of Theorem $1$ in \cite{Yang:jsac:2016}, 
\begin{align} \label{eqn:martingale-1}
\exp(-\alpha \Omega_0) = \Eb[\exp(-\alpha \Omega_{\kappa} - \kappa\gamma(\alpha))].
\end{align}
Taking the derivative of both sides of (\ref{eqn:martingale-1}) with respect to $\alpha$, 
\begin{align} \label{eqn:martingale-2}
\Omega_0 \exp(-\alpha \Omega_0) 
= \Eb[(\Omega_{\kappa}+\kappa\gamma^{\prime}(\alpha))\exp(-\alpha \Omega_{\kappa} - \kappa\gamma(\alpha) )]. 
\end{align}
Since $\Omega_0=1$ and $\Omega_{\kappa}=0$,  (\ref{eqn:martingale-2}) can be reduced to
\begin{align} \label{eqn:martingale-3}
\exp(-\alpha)
=\Eb[\kappa \gamma^{\prime}(\alpha) \exp(-\kappa \gamma(\alpha) )]\leq \Eb[\kappa\gamma^{\prime}(\alpha)],
\end{align}
where the inequality follows from the fact that $\kappa\gamma(\alpha) \geq 0$.

Dividing both sides of (\ref{eqn:martingale-3}) by $\gamma^{\prime}(\alpha) $, we have
\begin{align}
\Eb[\kappa]\geq \exp(-\alpha)/\gamma^{\prime}(\alpha).
\end{align}
Note that 
\begin{align}
\lim_{\alpha \rightarrow 0} \gamma^{\prime}(\alpha)&=\lim_{\alpha \rightarrow 0}(-e^{-\alpha}+1)=0^{+}.
\end{align} 
Thus, we have
\begin{align}
\lim_{T_0\rightarrow \infty}\Eb[T_1] \geq \lim_{\alpha \rightarrow 0} \exp(-\alpha)/\gamma^{\prime}(\alpha) =\infty.
\end{align}
\end{proof}

\begin{Lemma}\label{lemma:bounded}
Under BU-ER$_{T_0}$ updating, $\Eb[T_2]$, $\Eb[T^2_2]$, $\Eb[T_1-S_{N(T_1)}]$, $\Eb[(T_1-S_{N(T_1)})^2]$ are bounded.
\end{Lemma}

\begin{proof}
First, we note that under BU-ER$_{T_0}$ updating, the energy arrival over $[s_n,s_{n+1})$ is a Poisson random variable $A_{n+1}$ with parameter $1$. Therefore, if the battery level is zero at $s_n^+$, it remains zero at $s_{n+1}^+$ if $A_{n+1}=0$ or $1$, which happens with probability $q:=2e^{-1}$. It goes above one with probability $1-q$. Thus, $T_2$ is a geometric random variable with parameter $1-q$, whose first and second moments are bounded.

Next, we note that under the BU-ER$_{T_0}$ updating, the AoI over $[0,T_1]$ is a renewal reward process, which resets to zero at $\{S_i\}_{i=1}^{N(T_1)}$. According to Proposition 3.4.6 in~\cite{ross:1996}, $\lim_{t\rightarrow\infty}\Eb[S_{N(t)}-t]$ is bounded. Therefore $\Eb[S_{N(T_1)}-T_1]$ is uniformly bounded for any $T_1$. Similarly, we can show that $\Eb[(S_{N(T_1)}-T_1)^2]$ is uniformly bounded.
\end{proof}

\begin{Theorem}\label{lemma:AoI}
As $T_0\rightarrow\infty$, the expected long-term average AoI under BU-ER$_{T_0}$ is upper bounded by $\frac{2-p}{2p}$.
\end{Theorem}
\begin{proof}
First, we note that the 
\begin{align}
&\lim_{T_0\rightarrow \infty}\frac{\Eb[ (T_1+T_2-S_{N(T_1)})^2 ]}{2\Eb[T_1+T_2]}\nonumber \\
&= \lim_{T_0\rightarrow \infty} \frac{\Eb[(T_1-S_{N(T_1)})^2] +
	\Eb[T_2^2] +2 \Eb[T_1-S_{N(T_1)}]\Eb[T_2] }{2\Eb[T_1]} \label{eqn:usebounded-1} \\
&=0, \label{eqn:usebounded-2} 
\end{align}
where (\ref{eqn:usebounded-1}) follows from that the two events $T_1-S_{N(T_1)}$ and $T_2$ are independent, and (\ref{eqn:usebounded-2}) follows from Lemma \ref{lemma:N(T1)} and Lemma \ref{lemma:bounded}.

Then, we note that under BU-ER$_{T_0}$
\begin{align*}
\lim_{T\rightarrow\infty}\Eb\left[\frac{R(T)}{T}\right]&\leq \frac{\sum_{i=1}^{N(T_1) }X_i^2 + (T_1+T_2-S_{N(T_1)})^2 }{2\Eb[T_1+T_2]}.
\end{align*}

Consider the channel state realization at the scheduled status updating epochs under BU (and BU-ER) updating. Let $Y_i$ be the duration between the $i$th and $i-1$st epochs when the channel states are good and the corresponding update would be successful if it were sent. Then, $\{Y_i\}_{i=1}^{N(T_1)}$ is identical to $\{X_i\}_{i=1}^{N(T_1)}$. This is because there is no battery outage over $[0, T_1]$, and whether an update is successful or not only depends on the channel state.
Combining with (\ref{eqn:usebounded-2}), we have
 \begin{align}
&\lim_{T_0\rightarrow\infty}\lim_{T\rightarrow\infty}\Eb\left[\frac{R(T)}{T}\right]\leq \lim_{T_0\rightarrow\infty} \frac{\Eb[\sum_{i=1}^{N(T_1)}X_i^2 ]}{2\Eb[T_1+T_2]}  \label{eqn:upbound-01} \\
&\leq \lim_{T_0\rightarrow\infty} \frac{\Eb\left[\sum_{i=1}^{N(T_1)+1}Y_i^2 \right]}{2\Eb\left[\sum_{i=1}^{N(T_1)+1}Y_i-(\sum_{i=1}^{N(T_1)+1}Y_i-T_1)\right]}  \\
&=\lim_{T_0\rightarrow\infty} \frac{\Eb[N(T_1)+1]\Eb[Y_1^2]}{2\Eb[N(T_1)+1]\Eb[Y_1]-2\Eb\left[\sum_{i=1}^{N(T_1)+1}Y_i-T_1\right]}, \label{eqn:upbound-1} 
\end{align}
where (\ref{eqn:upbound-1}) follows from Wald's equality and the fact that $N(T_1)+1$ is a stopping time for $\{Y_i\}$ for any given $T_1$.

Since $\Eb[N(T_1)+1]\Eb[Y_1]\geq \Eb[T_1]$, according to Lemma~\ref{lemma:N(T1)}, 
\begin{align}
\lim_{T_0\rightarrow\infty}\Eb[N(T_1)+1]\Eb[Y_1]\geq \lim_{T_0\rightarrow\infty}\Eb[T_1]=\infty.
\end{align}
Meanwhile, we have $\Eb\left[\sum_{i=1}^{N(T_1)+1}Y_i-T_1\right]$ uniformly bounded for any $T_1$
based on Proposition 3.4.6 in \cite{ross:1996}. Therefore, (\ref{eqn:upbound-1}) is equal to $\frac{\Eb[Y_1^2]}{2\Eb[Y_1]}$, i.e., $\frac{2-p}{2p}$.
\end{proof}

Lemma \ref{lemma:lower}, Lemma \ref{lemma:suboptimal} and Theorem~\ref{lemma:AoI} imply the optimality of the BU updating, as summarized in the following theorem.

\begin{Theorem}\label{thm:myopic}
Among the set of bounded policies, the BU updating policy is optimal when the battery size is infinite. 
\end{Theorem}

\section{Simulation Results} \label{sec:simulation}
In this section, we evaluate the performances of the proposed status updating policies through simulations.

First, we generate sample paths for the Poisson energy harvesting process with $\lambda=1$, and perform the BU updating. The time average AoI as a function of $T$ is shown in Fig. \ref{fig:simu-1}. We vary $p=0.2, 0.6, 1.0$, and plot both the sample average and the corresponding lower bound over $500$ sample paths in the figure. We observe that all curves gradually approach the lower bound $\frac{2-p}{2p}$ as $T$ increases.  
We note that when $T =5000$, there is only a very small difference between the simulation results and the analytical lower bound. The results indicate that the proposed BU status updating policy is optimal. We also note that the time average AoI is monotonically decreasing as $p$ increases, which is consistent with the form of the lower bound. This is also intuitive since channel with better quality (i.e., larger $p$) will render smaller time-average AoI.

Next, we compare the time average AoI under different updating policies, i.e., the BU updating policy, the BU-ER$_{T_0}$ updating policy and a  greedy updating policy. We set $p=0.6$ and $T_0=30$. In the greedy updating policy, the transmitter updates instantly when one unit of energy arrives. We plot the sample average of AoI over $500$ sample paths in Fig.~\ref{fig:simu-2}. As we observe, the BU updating policy achieves the minimum time average AoI among those three updating policies. We note that the BU-ER updating policy also achieves the lower bound asymptotically, which is consistent with our proof, and the greedy policy does not approach the lower bound.

\begin{figure}[t]
	\centering
	\includegraphics[width=3.35in]{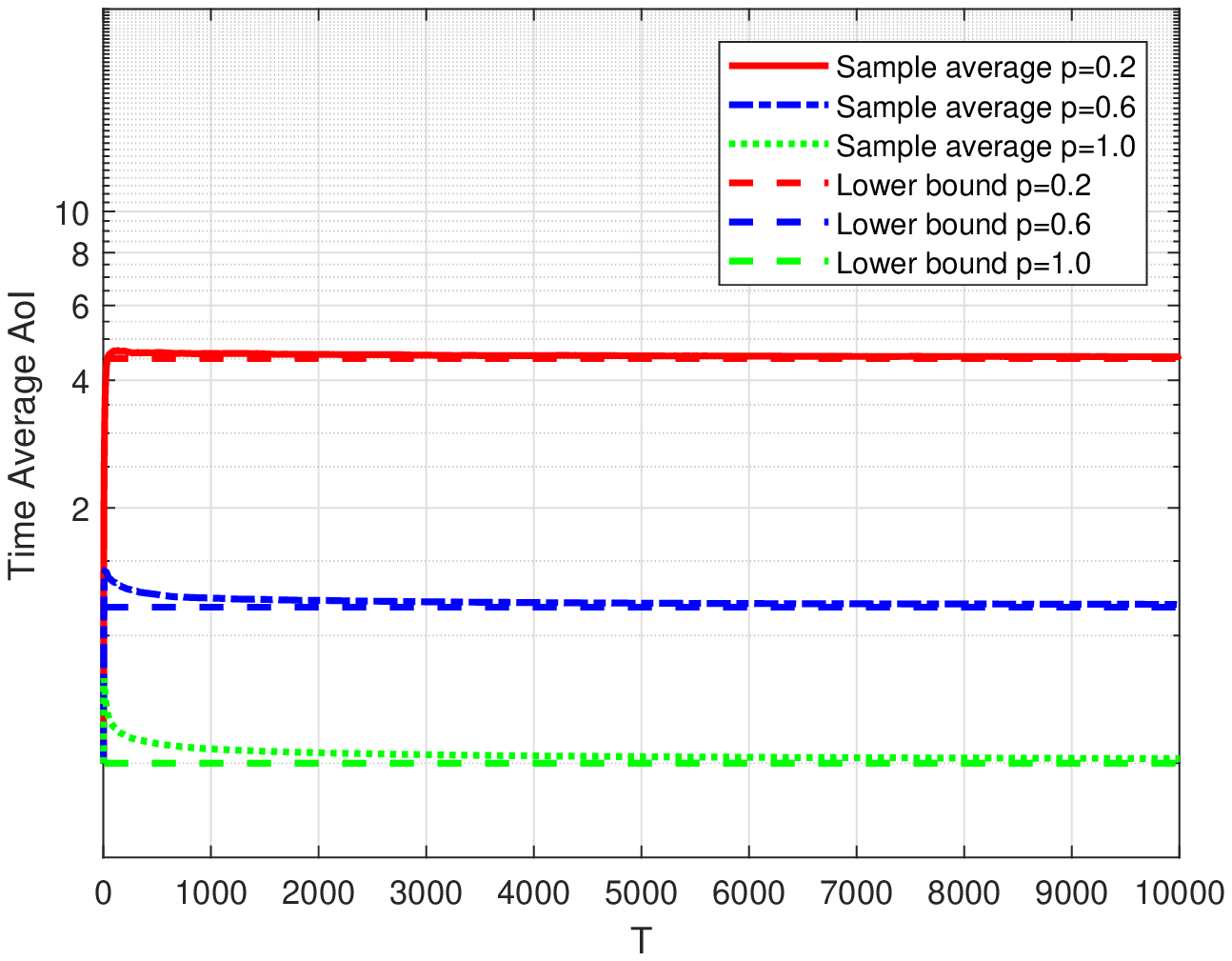}
	\caption{\mbox{Time average AoI with different $p$.}
		 }\label{fig:simu-1}
		 \vspace{-0.2in}
\end{figure}

\begin{figure}[t]
	\centering
	\includegraphics[width=3.35in]{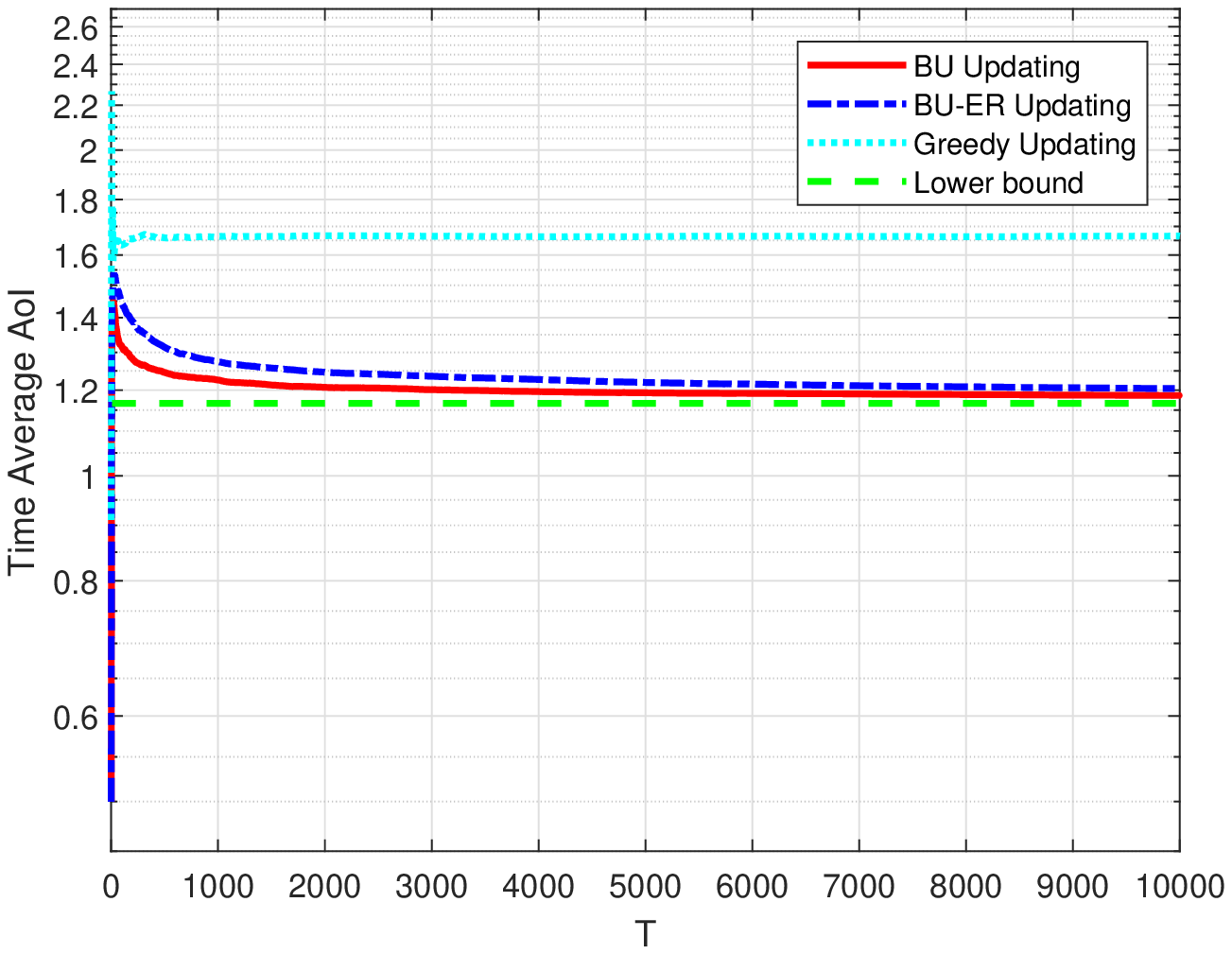}
	\caption{\mbox{Performance comparison.}
	}\label{fig:simu-2}
	 \vspace{-0.2in}
\end{figure}

\section{Conclusions} \label{sec:conclusion}
In this paper, we investigated the optimal status updating policy for an energy harvesting source with a noisy channel. We showed that among a broadly defined class of online policies, the BU updating policy minimizes the expected long-term average AoI. Its optimality is established by constructing a sequence of BU-ER updating policies which are sub-optimal to BU updating, and showing that its limit achieves the lower bound of the expected long-term average AoI.





\begin{thebibliography}{10}
	\providecommand{\url}[1]{#1}
	\csname url@samestyle\endcsname
	\providecommand{\newblock}{\relax}
	\providecommand{\bibinfo}[2]{#2}
	\providecommand{\BIBentrySTDinterwordspacing}{\spaceskip=0pt\relax}
	\providecommand{\BIBentryALTinterwordstretchfactor}{4}
	\providecommand{\BIBentryALTinterwordspacing}{\spaceskip=\fontdimen2\font plus
		\BIBentryALTinterwordstretchfactor\fontdimen3\font minus
		\fontdimen4\font\relax}
	\providecommand{\BIBforeignlanguage}[2]{{%
			\expandafter\ifx\csname l@#1\endcsname\relax
			\typeout{** WARNING: IEEEtran.bst: No hyphenation pattern has been}%
			\typeout{** loaded for the language `#1'. Using the pattern for}%
			\typeout{** the default language instead.}%
			\else
			\language=\csname l@#1\endcsname
			\fi
			#2}}
	\providecommand{\BIBdecl}{\relax}
	\BIBdecl
	
	\bibitem{Yang_tcom}
	J.~Yang and S.~Ulukus, ``Optimal packet scheduling in an energy harvesting
	communication system,'' \emph{{IEEE} Trans. Commun.}, vol.~60, no.~1, pp.
	220--230, Jan. 2012.
	
	\bibitem{ozel_finite_tcom}
	O.~Ozel, J.~Yang, and S.~Ulukus, ``Optimal broadcast scheduling for an energy
	harvesting rechargeable transmitter with a finite capacity battery,''
	\emph{{IEEE} Trans. Wireless Commun.}, vol.~11, no.~6, pp. 2193--2203, Jun.
	2012.
	
	\bibitem{HoZ12}
	C.~K. Ho and R.~Zhang, ``Optimal energy allocation for wireless communications
	with energy harvesting constraints,'' \emph{{IEEE} Trans. Signal Process.},
	vol.~60, no.~9, pp. 4808--4818, 2012.
	
	\bibitem{kaya_tcom}
	K.~Tutuncuoglu and A.~Yener, ``Optimum transmission policies for battery
	limited energy harvesting nodes,'' \emph{{IEEE} Trans. Wireless Commun.},
	vol.~11, no.~3, pp. 1180--1189, Mar. 2012.
	
	\bibitem{infocom/KaulYG12}
	S.~K. Kaul, R.~D. Yates, and M.~Gruteser, ``Real-time status: How often should
	one update?'' in \emph{{IEEE} {INFOCOM}}, Orlando, FL, USA, Mar. 2012, pp.
	2731--2735.
	
	\bibitem{ciss/KaulYG12}
	------, ``Status updates through queues,'' in \emph{Conference on Information
		Sciences and Systems (CISS)}, Princeton, NJ, USA, Mar. 2012, pp. 1--6.
	
	\bibitem{isit/YatesK12}
	R.~D. Yates and S.~K. Kaul, ``Real-time status updating: Multiple sources,'' in
	\emph{{IEEE} International Symposium on Information Theory (ISIT)},
	Cambridge, MA, USA, Jul. 2012, pp. 2666--2670.
	
	\bibitem{YatesK16}
	\BIBentryALTinterwordspacing
	------, ``The age of information: Real-time status updating by multiple
	sources,'' \emph{ArXiv e-prints}, 2016. [Online]. Available:
	\url{http://arxiv.org/abs/1608.08622}
	\BIBentrySTDinterwordspacing
	
	\bibitem{Pappas:2015:ICC}
	N.~Pappas, J.~Gunnarsson, L.~Kratz, M.~Kountouris, and V.~Angelakis, ``Age of
	information of multiple sources with queue management,'' in \emph{IEEE
		International Conference on Communications (ICC)}, Jun. 2015, pp. 5935--5940.
	
	\bibitem{isit/NajmN16}
	E.~Najm and R.~Nasser, ``Age of information: The gamma awakening,'' in
	\emph{{IEEE} International Symposium on Information Theory (ISIT)},
	Barcelona, Spain, Jul. 2016, pp. 2574--2578.
	
	\bibitem{isit/KamKNWE16}
	C.~Kam, S.~Kompella, G.~D. Nguyen, J.~E. Wieselthier, and A.~Ephremides, ``Age
	of information with a packet deadline,'' in \emph{{IEEE} International
		Symposium on Information Theory (ISIT)}, Barcelona, Spain, Jul. 2016, pp.
	2564--2568.
	
	\bibitem{isit/ChenH16}
	K.~Chen and L.~Huang, ``Age-of-information in the presence of error,'' in
	\emph{{IEEE} International Symposium on Information Theory (ISIT)},
	Barcelona, Spain, Jul. 2016, pp. 2579--2583.
	
	\bibitem{isit/KamKE13}
	C.~Kam, S.~Kompella, and A.~Ephremides, ``Age of information under random
	updates,'' in \emph{{IEEE} International Symposium on Information Theory
		(ISIT)}, Istanbul, Turkey, Jul. 2013, pp. 66--70.
	
	\bibitem{isit/KamKE14}
	------, ``Effect of message transmission diversity on status age,'' in
	\emph{{IEEE} International Symposium on Information Theory (ISIT)}, Honolulu,
	HI, USA, Jun. 2014, pp. 2411--2415.
	
	\bibitem{tit/KamKNE16}
	C.~Kam, S.~Kompella, G.~D. Nguyen, and A.~Ephremides, ``Effect of message
	transmission path diversity on status age,'' \emph{{IEEE} Trans. Inf.
		Theory}, vol.~62, no.~3, pp. 1360--1374, Mar. 2016.
	
	\bibitem{isit/CostaCE14}
	M.~Costa, M.~Codreanu, and A.~Ephremides, ``Age of information with packet
	management,'' in \emph{{IEEE} International Symposium on Information Theory
		(ISIT)}, Honolulu, HI, USA, Jun. 2014, pp. 1583--1587.
	
	\bibitem{tit/CostaCE16}
	------, ``On the age of information in status update systems with packet
	management,'' \emph{{IEEE} Trans. Inf. Theory}, vol.~62, no.~4, pp.
	1897--1910, Apr. 2016.
	
	\bibitem{isit/HuangM15}
	L.~Huang and E.~Modiano, ``Optimizing age-of-information in a multi-class
	queueing system,'' in \emph{{IEEE} International Symposium on Information
		Theory (ISIT)}, Hong Kong, China, Jun. 2015, pp. 1681--1685.
	
	\bibitem{isit/BedewySS16}
	A.~M. Bedewy, Y.~Sun, and N.~B. Shroff, ``Optimizing data freshness,
	throughput, and delay in multi-server information-update systems,'' in
	\emph{{IEEE} International Symposium on Information Theory (ISIT)},
	Barcelona, Spain, Jul. 2016, pp. 2569--2573.
	
	\bibitem{infocom/SunUYKS16}
	Y.~Sun, E.~Uysal{-}Biyikoglu, R.~D. Yates, C.~E. Koksal, and N.~B. Shroff,
	``Update or wait: How to keep your data fresh,'' in \emph{{IEEE INFOCOM}},
	San Francisco, CA, USA, Apr. 2016, pp. 1--9.
	
	\bibitem{Sun:ISIT:2017}
	Y.~Sun, Y.~Polyanskiy, and E.~Uysal-Biyikoglu, ``Remote estimation of the
	wiener process over a channel with random delay,'' in \emph{IEEE
		International Symposium on Information Theory (ISIT)}, Jun. 2017, pp.
	321--325.
	
	\bibitem{isit/Yates15}
	R.~D. Yates, ``Lazy is timely: Status updates by an energy harvesting source,''
	in \emph{{IEEE} International Symposium on Information Theory (ISIT)}, Hong
	Kong, China, Jun. 2015, pp. 3008--3012.
	
	\bibitem{ita/BacinogluCU15}
	B.~T. Bacinoglu, E.~T. Ceran, and E.~Uysal{-}Biyikoglu, ``Age of information
	under energy replenishment constraints,'' in \emph{Information Theory and
		Applications Workshop}, San Diego, CA, USA, Feb. 2015, pp. 25--31.
	
	\bibitem{Yang:2017:AoI}
	X.~Wu, J.~Yang, and J.~Wu, ``Optimal status update for age of information
	minimization with an energy harvesting source,'' \emph{IEEE Trans. on Green
		Communications and Networking}, 2018.
	
	\bibitem{BacinogluU17}
	B.~T. Bacinoglu and E.~Uysal{-}Biyikoglu, ``Scheduling status updates to
	minimize age of information with an energy harvesting sensor,'' \emph{CoRR},
	vol. abs/1701.08354, 2017.
	
	\bibitem{Ahmed:2018:ICC}
	A.~Arafa, J.~Yang, and S.~Ulukus, ``Age-minimal online policies for energy
	harvesting sensors with random battery recharges,'' in \emph{IEEE
		International Conference on Communications (ICC)}, May 2018.
	
	\bibitem{Ahmed:2018:ITA}
	A.~Arafa, J.~Yang, S.~Ulukus, and H.~V. Poor, ``Age-minimal online policies for
	energy harvesting sensors with incremental battery recharges,'' in
	\emph{Information Theory and Applications Workshop}, San Diego, CA, USA, Feb.
	2018.
	
	\bibitem{Uysal:2018:EH}
	\BIBentryALTinterwordspacing
	B.~{Tan Bacinoglu}, Y.~{Sun}, E.~{Uysal-Biyikoglu}, and V.~{Mutlu},
	``{Achieving the Age-Energy Tradeoff with a Finite-Battery Energy Harvesting
		Source},'' \emph{ArXiv e-prints}, Feb. 2018. [Online]. Available:
	\url{http://arxiv.org/abs/1802.04724}
	\BIBentrySTDinterwordspacing
	
	\bibitem{Ahmed:2017:Asilomar}
	A.~Arafa and S.~Ulukus, ``Age minimization in energy harvesting communications:
	Energy-controlled delays,'' in \emph{IEEE Asilomar}, Oct. 2017.
	
	\bibitem{Ahmed:2017:Globecom}
	------, ``Age-minimal transmission in energy harvesting two-hop networks,'' in
	\emph{IEEE Global Communications Conference}, Dec. 2017.
	
	\bibitem{Baknina:2018:CISS}
	A.~{Baknina} and S.~{Ulukus}, ``Coded status updates in an energy harvesting
	erasure channel,'' in \emph{Conference on Information Sciences and Systems
		(CISS)}, Mar. 2018.
	
	\bibitem{Baknina:2018:ISIT}
	\BIBentryALTinterwordspacing
	A.~{Baknina}, O.~{Ozel}, J.~{Yang}, S.~{Ulukus}, and A.~{Yener}, ``Sending
	information through status updates,'' \emph{ArXiv e-prints}, Jan. 2018.
	[Online]. Available: \url{http://arxiv.org/abs/1801.04907}
	\BIBentrySTDinterwordspacing
	
	\bibitem{Yang:jsac:2016}
	J.~Yang, X.~Wu, and J.~Wu, ``Optimal online sensing scheduling for energy
	harvesting sensors with infinite and finite batteries,'' \emph{{IEEE} J. Sel.
		Areas Commun.}, vol.~34, no.~5, pp. 1578--1589, May 2016.
	
	\bibitem{ross:1996}
	S.~Ross, \emph{Stochastic Processes}, ser. Wiley series in probability and
	statistics: Probability and statistics.\hskip 1em plus 0.5em minus
	0.4em\relax Wiley, 1996.
	
\end{thebibliography}

\end{document}